\documentclass[a4paper,11pt]{article}

\usepackage{amsthm}
\newtheorem{theorem}{Theorem}
\newtheorem{lemma}[theorem]{Lemma}
\newtheorem{obs}[theorem]{Observation}

\usepackage{color}

\usepackage{graphicx}
\usepackage{amsmath,amsfonts,amssymb}
\usepackage[utf8]{inputenc}
\usepackage{subfigure,multirow}
\usepackage{mathtools}

\newcommand{\remove}[1]{}

\usepackage{hyperref}
\usepackage{authblk}
\title{On the minimum spanning tree problem in imprecise set-up\thanks{A preliminary version of this paper appeared in COCOON, 2018}}
\author[1]{Sanjana Dey}
\author[2]{Ramesh K. Jallu\thanks{Ramesh K. Jallu  was supported by the Czech Science Foundation, grant number GJ19-06792Y, and by institutional support RVO:67985807}}

\author[1]{Subhas C. Nandy}
\affil[1]{Indian Statistical Institute, Kolkata, India}
\affil[2]{Institute of Computer Science, The Czech Academy of Sciences}
\date{}

\begin{document}

\maketitle

\begin{abstract}
In this article, we study the Euclidean minimum spanning tree problem in an imprecise set-up. 
The problem is known as the \emph{Minimum Spanning Tree Problem with Neighborhoods} in the literature. We study the problem where the neighborhoods are represented as non-crossing line segments. Given a set ${\cal S}$ of $n$ disjoint line segments in $I\!\!R^2$, the objective is to find  
a minimum spanning tree (MST) that contains exactly one end-point from each segment in $\cal S$ and the cost of the MST is minimum among $2^n$ possible MSTs. We show that finding such an MST 
is NP-hard in general, and propose a $2\alpha$-factor approximation algorithm 
for the same, where $\alpha$ is the approximation factor of the best-known 
approximation algorithm to compute a minimum cost Steiner tree in an undirected 
graph with non-negative edge-weights. As an implication of our reduction, we can show that the unrestricted 
version of the problem (i.e., one point must be chosen from each segment such that the cost of 
MST is as minimum as possible) is also NP-hard.
We also propose a parametrized algorithm for the problem based on the ``separability'' parameter defined for segments.

\end{abstract}


\section{Introduction} 
Points are basic objects when dealing with geometric optimization problems. In Computational geometry, most of the algorithms dealing with point objects assume that the co-ordinates of the points are precise.
However, it may not always be possible to have the precise or exact co-ordinates (or locations) of these points. For example, the point locations obtained from a navigation system such as GPS may not always be precise due to several reasons like signal obstruction, multi-path effects, etc. In such cases, the classical geometric algorithms might fail to produce correct solutions with imprecise data. The study has drawn much attention when dealing with geometric optimization problems associated with imprecise data.
In such problems, generally, the input instances are provided as regions (also known as regions of uncertainty) such as a set of points, disks, rectangles, line segments, etc., and the objective is to choose a point from each region such that the cost function of some desired geometric structure (e.g., a minimum spanning tree, a convex hull, a traveling salesman tour, an enclosing circle, etc.) or a measure (e.g., diameter, pairwise distance, shortest path, bounding box, etc.) constructed on the chosen point set is as minimum or maximum as possible.

In this article, we study the Euclidean minimum spanning tree (MST) problem under the imprecision model represented as disjoint line segments. The objective is to find a point from each segment such that the cost of the MST constructed is as minimum as possible. In the rest of the paper, we use the terms cost and weight interchangeably 
to indicate the cost of an edge or tree where ever applicable.

\section{Related Work}
The \emph{MST problem} has been a well-studied problem in both graph and geometric 
domain for decades. In the context of an edge-weighted graph $G$, the objective is to find a 
tree spanning all the nodes in $G$ such that the sum of the weights of all 
the tree edges is minimized. In the geometric 
setup, the nodes of the underlying graph correspond to a given set of objects, 
the graph is complete, and each edge of the graph is the distance (in some 
appropriate metric) between the objects corresponding to the nodes incident to 
that edge. If the objects are points in $I\!\!R^d$ and the distances are in the Euclidean metric, the problem 
is referred to as the {\it Euclidean MST} problem. In $I\!\!R^2$ and $I\!\!R^3$, the problem can be 
solved in $O(n\log n)$ and $O((n\log n)^{4/3})$ time, respectively. In $I\!\!R^d$, 
the best-known algorithm runs in sub-quadratic time \cite{agar1991}. 
For any given constant $\varepsilon >0$, it is possible to produce a 
$(1+\varepsilon)$-approximation for the Euclidean MST problem in $I\!\!R^d$ ($d \geq 2$)  
in $O(n \log n)$ time \cite{smid2016}. 
For a survey on the Euclidean 
MST problem for a given point set in $I\!\!R^d$, one can refer to \cite{BOSE2013818,Eppstein}.

A natural generalization of the spanning tree problem on a graph is the {\it Steiner tree problem} (STP), where the input is a simple undirected weighted graph with a distinguished subset of vertices, called \emph{terminal nodes}, and the
STP asks for a minimum cost tree spanning all the terminal nodes (which  may not span all the non-terminal nodes). The STP problem is one of the classical NP-hard problems, and it cannot be approximated within a factor of $\frac{96}{95}$ \cite{chlebik2008}. A series of approximation 
results (e.g., \cite{mehlhorn1988,zelikovsky1993,hougardy1999,promel2000,robins2005,byrka2010improved}) are available in the literature that improved approximation factors 2 to the current best $\ln 4 + \varepsilon$, where $\varepsilon >0$ is any constant. If the set of terminal nodes is the entire vertex set, then the minimum cost Steiner tree is nothing but the MST. The Euclidean STP problem is NP-hard even in $I\!\!R^2$, and unlike in graphs it admits polynomial-time approximation scheme (PTAS) \cite{arora1998}.

The MST problem under the imprecision model represented as disjoint unit disks was first studied in \cite{yang2007}. The problem has been shown to be NP-hard (later in \cite{dorrigiv2015}, it has been claimed that the hardness proof provided in \cite{yang2007} is erroneous and said to have given the correct proof) and admits a PTAS. For the regions of uncertainty modeled as disks or squares that are not necessarily pairwise disjoint, L{\"o}ffler and Kreveld \cite{loffler2010} showed that the problem is NP-hard. Dorrigiv et al. \cite{dorrigiv2015} studied the problem (together with it's maximization version, where the objective is to maximize the cost of MST) for disjoint unit disks and showed that the problems are NP-hard. Their reduction also ensures that both the versions do not admit an Fully PTAS, unless $P=NP$. The authors provided deterministic and parameterized approximation algorithms for the maximized version of the problem, and a parameterized approximation algorithm for the MST problem. In continuation, Disser et al. \cite{disser2014} studied the problem in $L_1$ metric with imprecision modeled as line segments, and showed that the problem is APX-hard. Recently, Mesrikhani et al. \cite{mesrikhani2019} studied both the versions of the problems for axis-aligned pairwise disjoint squares in the plane under $L_1$ metric, and proposed constant factor and parametrized approximation algorithms for the maximization of the problem.

The study of geometric optimization problems in imprecise set-up is not just restricted to the MST problem, but studied for many other important problems such as convex hull \cite{loffler2010}, traveling sales man problem \cite{arkin1994,dumitrescu2003,de2005tsp,elbassioni2009,mitchell2010}, shortest path \cite{loffler2010largest, disser2014}, enclosing circle, smallest diameter, largest bounding box \cite{loffler2010largest}. As described earlier, for all the problems mentioned above the input is provided by a set of regions (either continuous or discrete, and either disjoint or non-disjoint). From each region, one point must be selected such that the respective problem's objective is met for the chosen set of points.

\subsection{Our work}
We first concentrate on a different variation of the MST problem for a 
set $\cal S$ of non-crossing line segments in $I\!\!R^2$. The objective is to find  
an MST that contains exactly one end-point from each segment in $\cal S$ and the cost of the MST is as minimum as possible (see Figure \ref{fig:msts}). 
\begin{figure}[!h]
\centering
\includegraphics[scale=0.6]{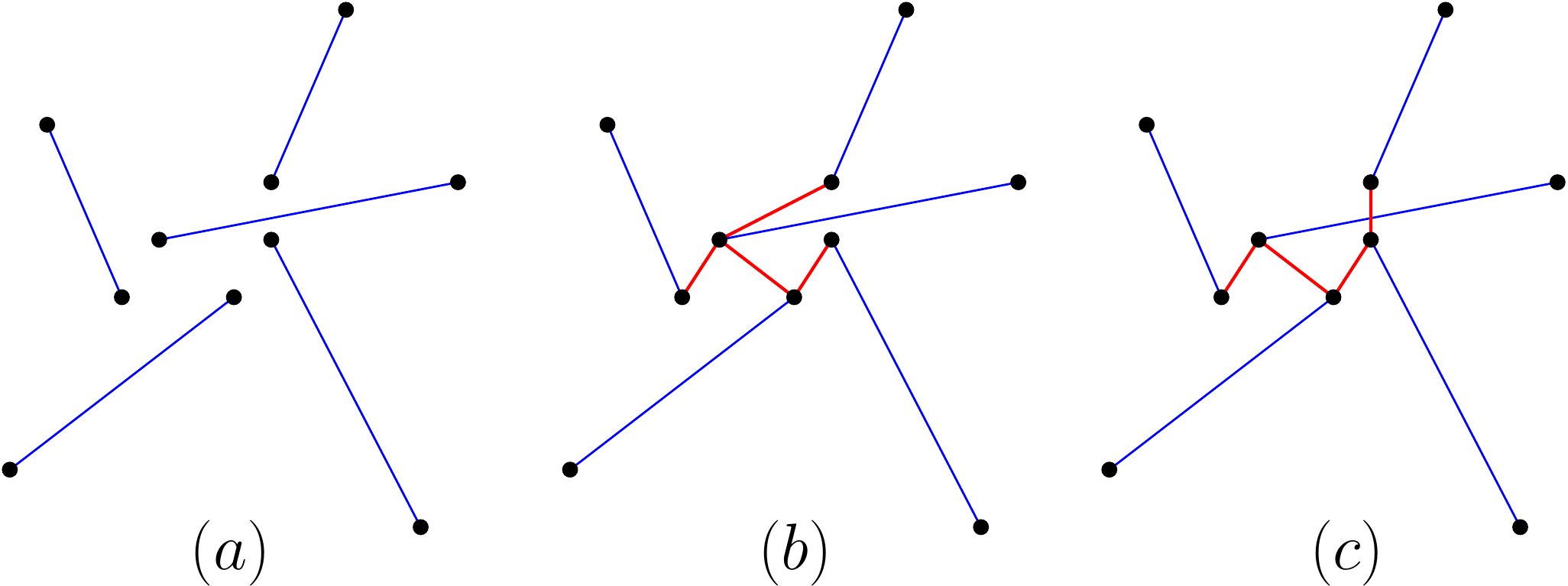}
\caption{(a) A set of line segments in the plane, (b) a spanning tree of the 
segments, and (c) a minimum spanning tree}\label{fig:msts}
\end{figure}

From now on, we will 
refer this problem as the {\em minimum spanning tree of 
segments} (MSTS) problem. Surely, if the segments in ${\cal S}$ are of 
length zero, then the MSTS problem reduces to the standard Euclidean MST problem.
We show that in the non-degenerate case, i.e., where not all the segments are of length zero, 
the MSTS problem is NP-hard (refer Section \ref{sec:msts_nphard}). As an implication of the reduction, we can show 
that the unrestricted MSTS problem (formally defined in Subsection \ref{sub_sec:implication}) is NP-hard. In Section \ref{sec:para_algo}, 
we propose a parametrized algorithm for the problem with 
the separability of the segments in ${\cal S}$ as the parameter, where the concept of separability is introduced in \cite{dorrigiv2015}.
A more general problem in this context is the {\it generalized minimum spanning tree} (GMST) problem for an edge weighted graph $G=(V,E)$, 
and a cost function $c: E \rightarrow I\!\!R^+$, 
where $V=\cup_{i=1}^m V_i$ and $V_i \cap V_j=\emptyset$ for all $i, j \in \{1,2,\ldots,m\}$ such that $i \neq j$; the objective is to choose exactly one vertex from each $V_i$ such that their spanning tree is of minimum cost.
In \cite{myung1995}, it was shown that the GMST problem is NP-hard, and no polynomial-time heuristic algorithm with a finite worst-case performance ratio can exist, unless $P = NP$. If $|V_i| \leq \rho$ 
(a constant), for all $i=1,2,\ldots, m$, and the cost function satisfies triangle inequality, then using  linear programming relaxation of the integer programming formulation of the problem, an algorithm can be designed with approximation factor
$(2-\frac{2}{n})\rho$ \cite{pop}. Our MSTS problem is the Euclidean version of the GMST problem with $\rho=2$, where the edge costs are Euclidean distances. Thus, the result of \cite{pop} suggests a straightforward 4-factor approximation algorithm for the MSTS problem. In Section \ref{sec:approx}, we propose an improved approximation algorithm which yields an MST whose cost is within a factor of $2\alpha$ of the optimum MST, where $\alpha$ is the best-known approximation factor for the Steiner tree problem for undirected graphs. In \cite{byrka2010improved}, it was shown that the value of $\alpha = \ln 4 + \varepsilon < 1.39$, where $\varepsilon > 0$ is any constant. Thus, we get a 2.78-factor approximation algorithm for the MSTS problem and also for the GMST problem considered in \cite{pop} .

A closely related problem in our context is the 2-Generalized minimum spanning tree (2-GMST) problem, in which each imprecise vertex contains 
exactly two points in the plane having the same $x$-coordinates. In \cite[p.~137]{fraser2013}, the 2-GMST problem is shown 
to be NP-hard and that no fully PTAS exists for the problem, unless $P = NP$. This reduction can be adopted to prove the MSTS problem NP-hard. However, to the best of our understanding, this reduction cannot be adopted to prove the NP-hardness of the unrestricted version of the MSTS problem. So, 
we provide a separate reduction for proving the unrestricted version of the MSTS problem to be NP-hard. It needs to be mentioned that, our hardness result is for a special case in which the segments are pairwise disjoint.

\section{The MSTS problem is NP-hard}\label{sec:msts_nphard}
We prove the MSTS problem is NP-hard by a reduction from the \textsc{Max 2-Sat} problem, 
defined below, and which is known to be NP-complete \cite{garey2002}.
 
\begin{description}
 \item [Instance:] A Boolean formula consisting of $n$ variables and a 
 set $\{C_1,C_2,\ldots,C_m\}$ of $m$ clauses, each $C_i$ is a disjunction 
 of at most two literals, and an integer $k$, $1 \leq k < m$.
 \item [Question:] Is there a truth assignment to the variables that 
 simultaneously satisfies at least $k$ 
 clauses?
\end{description}

Given an instance of \textsc{Max 2-Sat}, we get an instance of the MSTS problem 
 such that the given \textsc{Max 2-Sat} 
formula satisfies $k$ clauses if and only if the cost of the MSTS attains a specified value.

We represent a horizontal line segment $s$ as $(l(s),r(s))$, where $l(s)$ and $r(s)$ are 
the left and right end-points of $s$, respectively. Similarly, a vertical line segment $s$ is 
represented as $(t(s),b(s))$, where $t(s)$ and $b(s)$ are the top and bottom end-points of 
$s$, respectively. For a point $p$ in the plane, we use $x(p)$ and $y(p)$ to denote the 
$x$- and $y$-coordinate of $p$, respectively.

Our reduction is similar to \cite{daescu2010}, however, the gadgets we use in our reduction
are entirely different. Let $\psi$ be a \textsc{2-Sat} formula having $m$ clauses $C_1,C_2,\ldots,C_m$ 
and $n$ variables $x_1,x_2,\ldots,x_n$. 
We use the following notation as used in \cite{daescu2010}.
Let $x_{i,j,k}$ (or $\neg x_{i,j,k}$) be 
the variable $x_i$ that appears at the $j$-th literal in $\psi$ from left to 
right such that $x_i$ (including $\neg x_i$) appears $k-1$ times already in $\psi$ 
before the current occurrence of $x_i$. 
For example in the Boolean formula $(x_1 \vee \neg x_2) \wedge (x_2 \vee x_3)$, 
the literals $x_1, \neg x_2,x_2$, and $x_3$ are represented as 
$x_{1,1,1},\neg x_{2,2,1},x_{2,3,2}$, and $x_{3,4,1}$, respectively.
We create gadgets for each variable $x_i\;(1 \leq i \leq n)$ and for each 
literal $x_{i,j,k}$ (or $\neg x_{i,j,k}$). Each gadget consists of a set of horizontal and vertical line segments in the plane.

\begin{figure}[h]
  \centering
  {\includegraphics[scale=1]{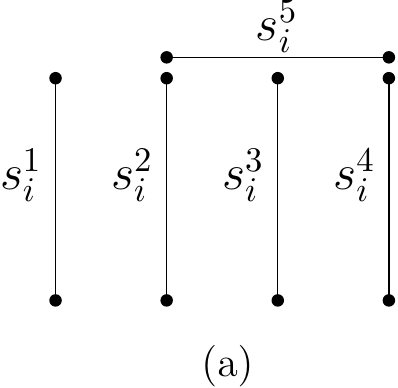}}\hspace{1cm}
  {\includegraphics[scale=1]{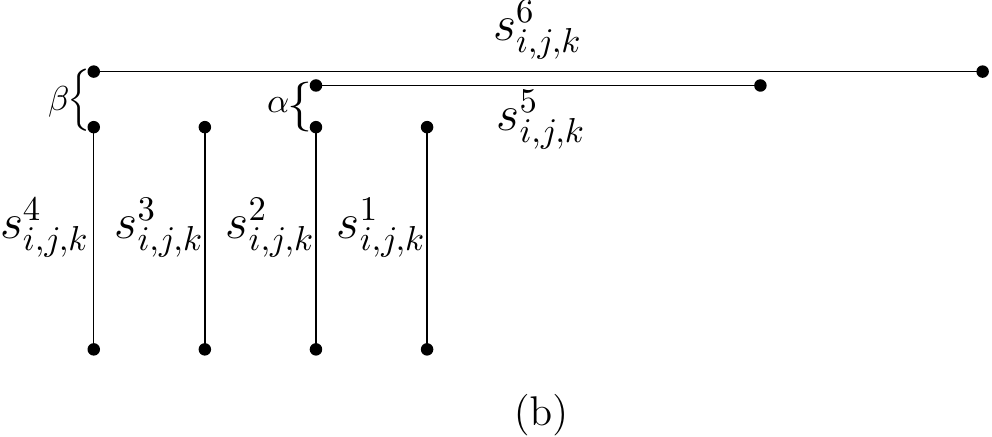}} 
  \caption{(a) Variable gadget for $x_i$, and (b) literal gadget for $x_{i,j,k}$}\label{fig:gadgets}
\end{figure}

{\bf Variable gadget:} For each variable $x_i$, five disjunct segments $s_i^l\;(1 \leq l \leq 5)$ 
are considered. The first four are vertical and the last one is horizontal. 
The vertical segments are of equal length (say $\lambda > 0$) and their top end-points are 
horizontally aligned with unit distance between consecutive end-points. The horizontal 
segment $s_i^5$ spans as follows: $x(l(s_i^5))=x(t(s_i^2)), y(l(s_i^5))=y(t(s_i^2))+\varepsilon, 
x(r(s_i^5))=x(t(s_i^4)), y(r(s_i^5))=y(t(s_i^2))+\varepsilon$,
where $\varepsilon$ is a very small positive real number (see Figure \ref{fig:gadgets}(a)).

{\bf Literal gadget:} For each literal $x_{i,j,k}$ (or $\neg x_{i,j,k}$),
six disjunct segments $s_{i,j,k}^l \;(1 \leq l \leq 6)$ are considered. The first four segments
are vertical, while the last two are horizontal (see Figure \ref{fig:gadgets}(b)). 
Here also the vertical segments are of same length and their top 
end-points are aligned as discussed in variable gadget.
The two horizontal segments are above the vertical segments and are of different lengths.
Their lengths depend on how many times its associated variable $x_i$ 
appears in different clauses of $\psi$. 
The precise co-ordinates for the left end-point of the horizontal segment $s_{i,j,k}^5$ (resp. $s_{i,j,k}^6$) are at $(x(t(s_{i,j,k}^2)), y(t(s_{i,j,k}^2))+\alpha)$
(resp. $(x(t(s_{i,j,k}^4)), y(t(s_{i,j,k}^4))+\beta)$), 
where $\alpha=2j\varepsilon$ (resp. $\beta=(2j+1)\varepsilon$). 
The two horizontal segments will be used to connect with the gadget of variable $x_i$.
In the connection, the right end-points of the horizontal segments are vertically aligned 
with $s_i^2$ and $s_i^4$ (see Figure \ref{fig:var-lit-gad}).

\begin{figure}[!h]
 \centering
 \includegraphics[scale=1]{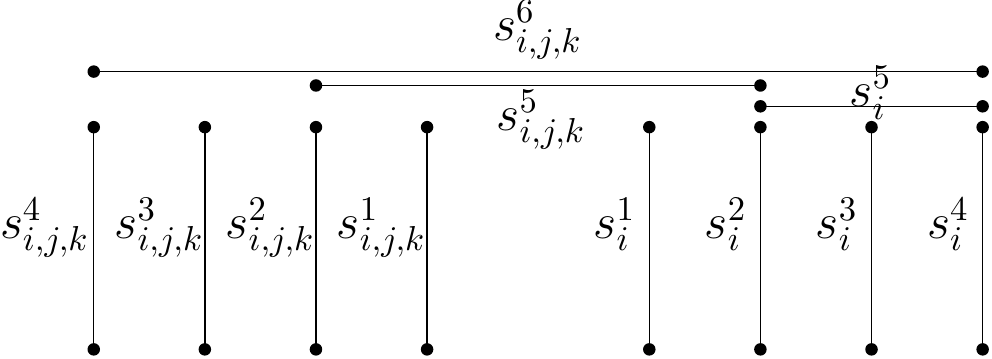}
 \caption{Connecting $x_{i,j,k}$'s gadget to $x_i$'s gadget}\label{fig:var-lit-gad}
\end{figure}

The basic idea of putting the segments in this manner is to get two minimum spanning 
trees of the segments for connecting the gadgets of $x_i$ and $x_{i,j,k}$ depending 
on whether $x_i$ is set to true or false respectively. If $x_i$ = {\it false}, then we 
choose the end-points 
$F=\{l(s_{i,j,k}^6)$, $t(s_{i,j,k}^4),t(s_{i,j,k}^3),t(s_{i,j,k}^2),t(s_{i,j,k}^1),t(s_i^1),t(s_i^2),
t(s_i^3),t(s_i^4)$, $l(s_i^5),r(s_{i,j,k}^5)\}$, and if 
$x_i$ = {\it true}, then we choose the end-points $T=\{t(s_{i,j,k}^4),t(s_{i,j,k}^3)$, $t(s_{i,j,k}^2),
t(s_{i,j,k}^1),t(s_i^1),t(s_i^2),t(s_i^3),t(s_i^4),l(s_{i,j,k}^5)$, $r(s_i^5),r(s_{i,j,k}^6)\}$.
In Figure \ref{fig:gad_mst} these two minimum spanning trees are shown in bold. 

\begin{figure}[!ht]
  \centering
  {\includegraphics[scale=.75]{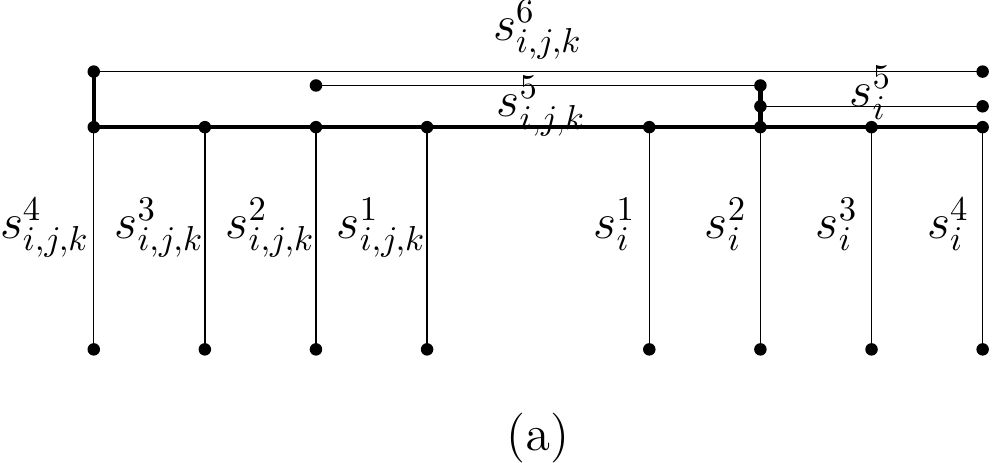}\label{fig:gad_mst2}}
  \hfill
  {\includegraphics[scale=.75]{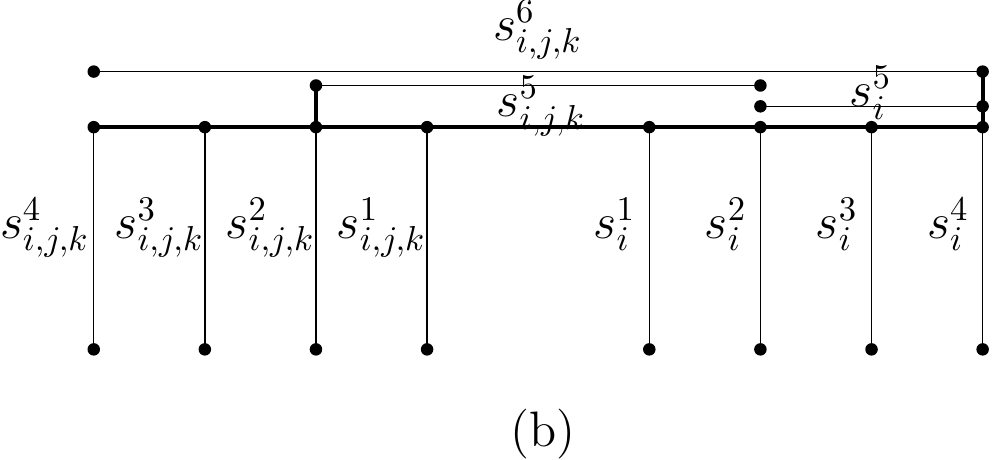}\label{fig:gad_mst1}}
  \caption{(a) The MSTS when $x_i$ is false, and (b) the 
  MSTS when $x_i$ true.}\label{fig:gad_mst}
\end{figure}

We now explain the arrangement of the gadgets according to the given formula $\psi$.
The gadgets corresponding to the variables $x_1,x_2,\ldots,x_n$ are placed on the positive part of the $x$-axis 
in left to right order with the top end-points of the vertical segments aligned with the $x$-axis. 
Similarly, the gadgets corresponding to the literals are arranged on the negative part of the 
$x$-axis from right to left (in the order they appear in $\psi$). Needless to mention, 
the $y$-coordinates of the top end-point of the vertical segments 
(in both variable and literal gadgets) are zero.
As mentioned earlier, the horizontal segments  of a literal gadget establish a connection
with its corresponding variable gadget\footnote{A variable gadget may be connected with 
multiple literal gadgets depending on how many times it appeared in $\psi$.}. 
Note that, the formula for the $y$-coordinates of the horizontal segments (in literal gadgets) ensure that no two horizontal 
segments will overlap. In Figure \ref{fig:example}, the arrangement of variable and literal gadgets 
with their connection is shown for $\psi = x_1 \vee x_2$.
\begin{figure}[!hb]
 \centering
 \includegraphics[scale=0.75]{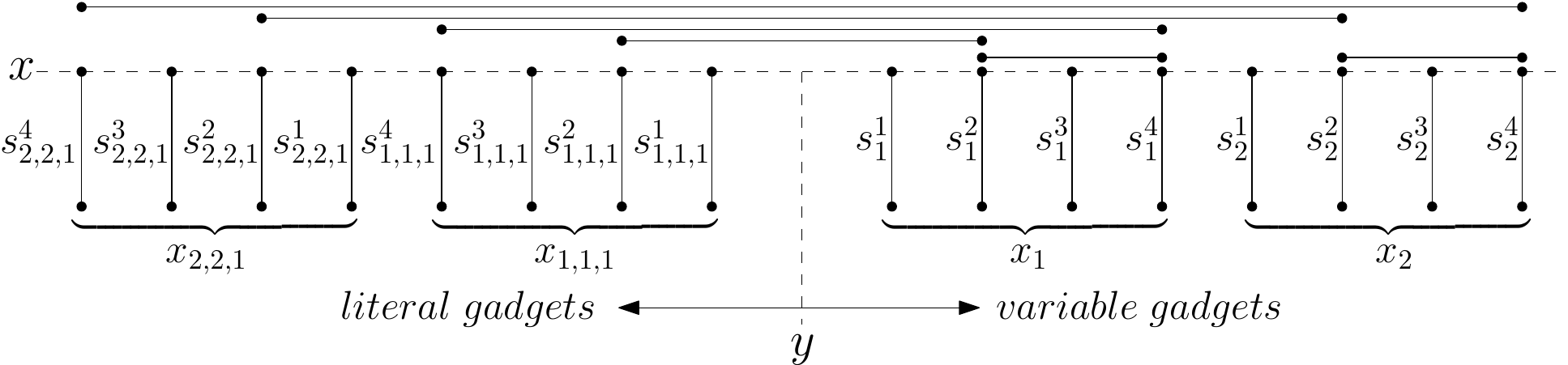}
 \caption{Placement of variable and literal gadgets in the 
 clause $x_1 \vee x_2$}\label{fig:example}
\end{figure}

According to our construction, we get a set ${\cal S}$ of $5n + 12m$ segments 
($5n$ segments correspond to $n$ variables and $12m$ segments correspond to $2m$ 
literals in $m$ clauses). Let ${\cal T}_{\cal S}$ be 
an MSTS over ${\cal S}$ and let ${\cal W}({\cal T}_{\cal S})$ be its weight. 

The tree ${\cal T}_{\cal S}$ satisfies the following properties: 
\begin{itemize}
\item[(i)] In every (variable and literal) 
gadget, the top end-points of all its vertical segments are part of it; 
\item[(ii)] No edge in ${\cal T}_{\cal S}$  
connects any two horizontal segments of literal gadgets;
\item[(iii)] No edge in ${\cal T}_{\cal S}$  
connects horizontal segments of any two variable gadgets;
and 
\item[(iv)] The horizontal segments of any literal $x_{i,j,k}$ 
(or $\neg x_{i,j,k}$) gadget are either connected  to its
vertical segments or segments in its associated variable $x_i$'s 
gadget using a vertical edge of ${\cal T}_{\cal S}$.
\end{itemize}
These properties justify the need of taking four equidistant 
vertical segments in each gadget as it prohibits 
from choosing non axis-parallel edges of ${\cal T}_{\cal S}$ 
to connect segments in the gadgets.

So far, we have discussed the arrangement of variable and literal 
gadgets and the nature of the minimum spanning tree 
of those segments for a true/false assignment of the $n$ variables. 
Now, we add one special horizontal segment\footnote{this segment corresponds to the binary 
relation \textbf{\textit{or}}.} for each clause in $\psi$. 
These segments appear in the same vertical level in which the 
horizontal segments of the variable gadgets appear, 
and the placement of the segment for a clause is as follows: for $i_1,i_2\in \{1,2,\ldots,n\}$, consider a clause $x_{i_1} \vee x_{i_2}$ in $\psi$ such that $x_{i_1}$ (resp. $x_{i_2}$) 
appears at the $j$-th (resp. ($j+1$)-st) literal in $\psi$ from left to 
right such that $x_{i_1}$ (resp. $x_{i_2}$) appears $k_1-1$ (resp. $k_2-1$) times already in $\psi$ 
before the current occurrence of $x_{i_1}$ (resp. $x_{i_2}$). That is, the literals $x_{i_1}$ and $x_{i_2}$ 
are represented as $x_{i_1,j,k_1}$ and $x_{i_2,j+1,k_2}$, respectively. 
We place a segment, say $s^{\lceil\frac{j}{2}\rceil}$, that connects the gadgets correspond 
to $x_{i_1,j,k_1}$ and $x_{i_2,j+1,k_2}$. The placement of the right and left end-points of 
the segment $s^{\lceil\frac{j}{2}\rceil}$ depends on how $x_{i_1}$ and $x_{i_2}$ appear (i.e., 
whether negative or positive) in $\psi$. If $x_{i_1}$ appears as 
positive (resp. negative) literal in $\psi$, then the \emph{right} end-point of the segment 
$s^{\lceil\frac{j}{2}\rceil}$ is at $\varepsilon$ vertical distance away from the top end-point of 
$s_{i_1,j,k_1}^2$ (resp. $s_{i_1,j,k_1}^4$). Similarly, if $x_{i_2}$ appears as 
positive (resp. negative) literal in $\psi$, then the \emph{left} end-point of the segment 
$s^{\lceil\frac{j}{2}\rceil}$ is at $\varepsilon$ vertical distance away from the top end-point of  
$s_{i_2,j+1,k_2}^2$ (resp. $s_{i_2,j+1,k_2}^4$).
\remove{Depending on whether the literal $x_{i,j,k}$ corresponds to the variable $x_i$ or its negation $\neg x_i$, 
we put the corresponding end-point of the segment 
$s^{\lceil\frac{j}{2}\rceil}$ on the top of 
$s_{i,j,k}^2$ or $s_{i,j,k}^4$, respectively.}
Figure \ref{fig:example1}(a) and Figure \ref{fig:example1}(b) demonstrate 
these new horizontal segments (with its end-points represented as black squares) 
for clauses $x_1 \vee x_2$ and $x_1 \vee \neg x_2$, respectively. 
Figure \ref{fig:2sat_ex} shows the MSTS problem instance for the Boolean formula 
$\psi = (\neg x_1 \vee x_2)\wedge (x_2 \vee \neg x_3)$.
 Let ${\cal S}' = {\cal S} \cup \{s^{\lceil\frac{j}{2}\rceil} \mid j=1,3,\ldots,2m-1\}$, 
 and let ${\cal T}$ be an MSTS of this extended set ${\cal S}'$ of 
 $5n + 13m$ segments with weight $W({\cal T})$.

\begin{figure}[!t]
 \centering
 \includegraphics[scale=0.75]{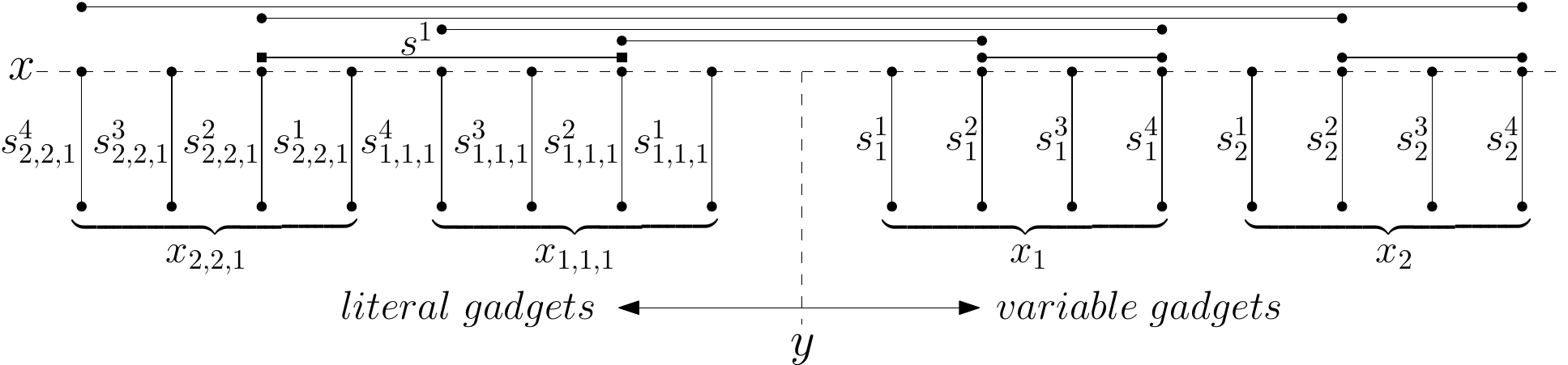} \\ (a) \\
  \includegraphics[scale=0.75]{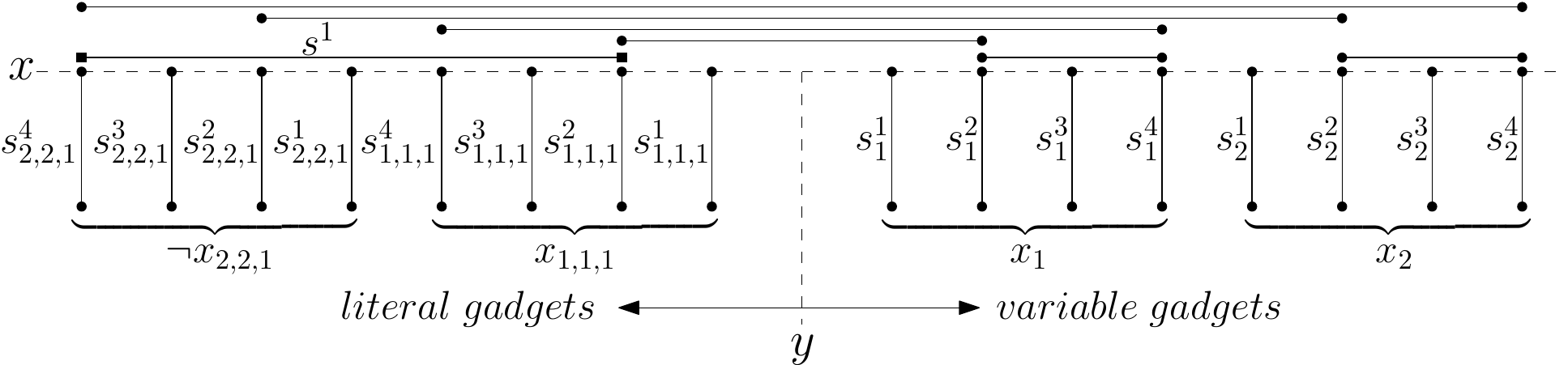} \\(b)
 \caption{Gadget for the clause (a) $x_1 \vee x_2$, and (b) $x_1 \vee \neg x_2$}\label{fig:example1}
\end{figure}

\begin{figure}[!h]
 \centering
 \includegraphics[scale=0.85]{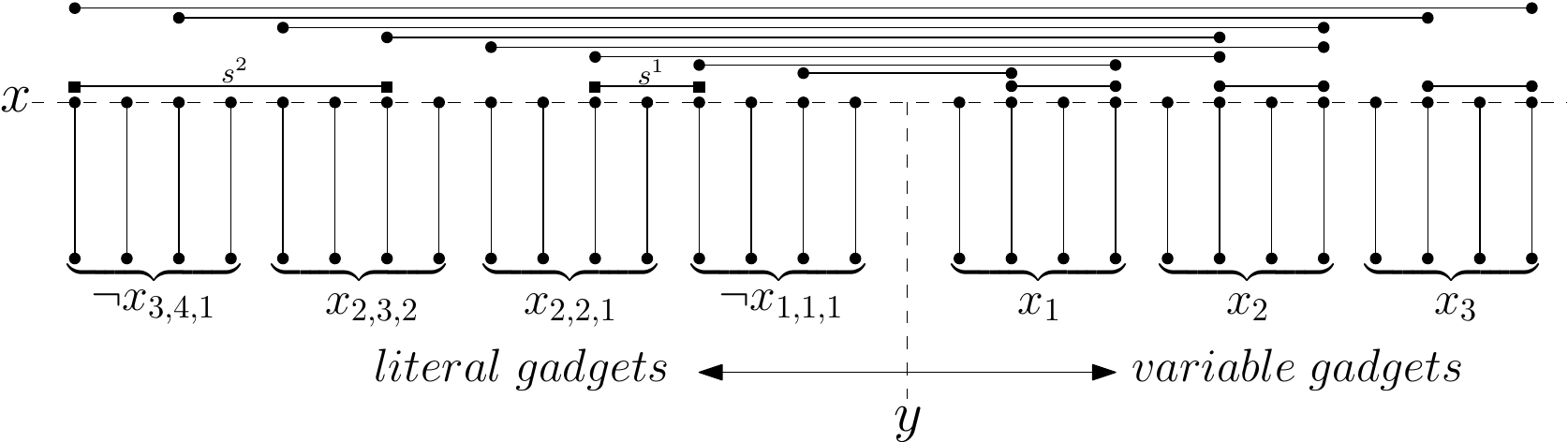}
 \caption{Gadget for $\psi = (\neg x_1 \vee x_2)\wedge (x_2 \vee \neg x_3)$}\label{fig:2sat_ex}
\end{figure}

\begin{lemma}\label{lem:msts_assign}
From an MSTS ${\cal T}$ of ${\cal S}'$ one can get a conflict-free assignment of the variables. 
\end{lemma}

\begin{proof}
Without loss of generality, we assume that the horizontal 
segments of the literal and variable gadgets are connected to 
vertical segments in ${\cal T}$, in one of the ways shown in 
Figure \ref{fig:gad_mst}. If there is a connection of a pair of segments in ${\cal T}$ which is not 
in either of the forms as shown in Figure \ref{fig:gad_mst} (see the curved edges in the left 
part of Figure \ref{fig:tree_alter}), then we can alter the tree without 
changing its weight (see the right part in Figure \ref{fig:tree_alter}).

\begin{figure}[!h]
 \centering
 \includegraphics[scale=0.75]{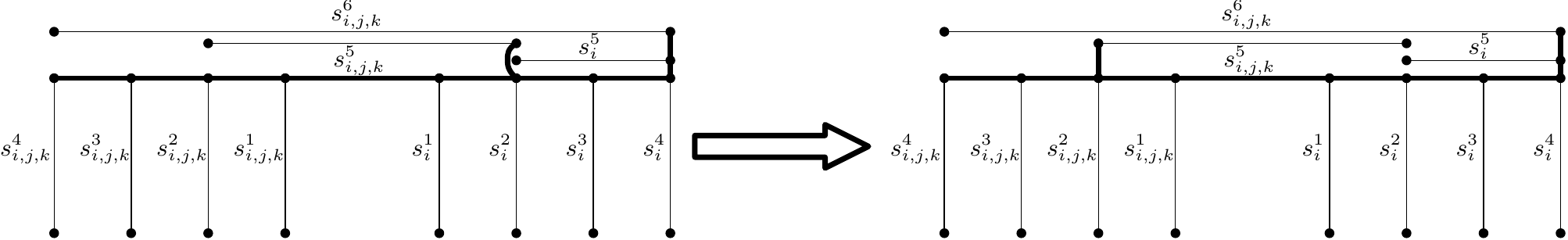} \\
 \vspace*{.5cm}
 \includegraphics[scale=0.75]{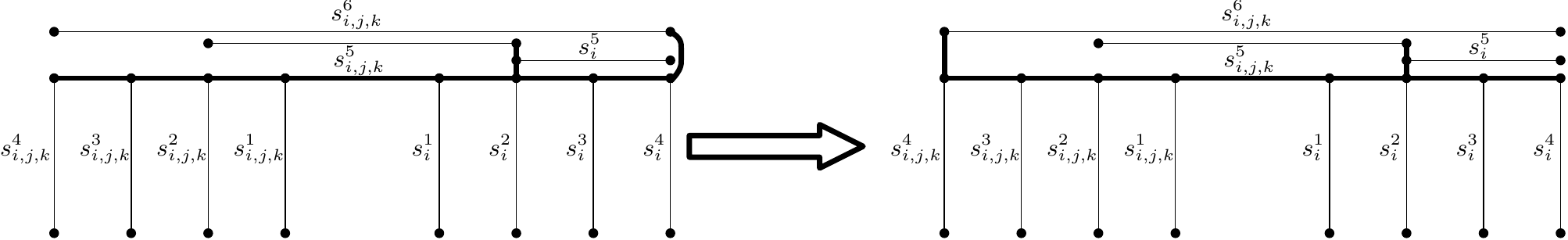} \\ 
 \caption{The possible tree alterations in an MSTS}\label{fig:tree_alter}
\end{figure}

For every variable $x_i$'s gadget,
we check which end-point between $t(s_i^2)$ and $t(s_i^4)$ is chosen to connect to its horizontal 
segment $s_i^5$. Note that, both the end-points 
cannot be connected to the horizontal segment simultaneously due to the feasibility\footnote{only 
one end-point of a segment can participate in the tree} of the problem. 
By our construction of the gadgets, if $t(s_i^2)$ is connected then we set $x_i ={\it false}$, else  
set $x_i ={\it true}$. 
Also note that, if a variable $x_i$ gets an assignment it never changes.
Because ${\cal T}$ is a spanning tree with minimum weight, if $t(s_i^2)$ (or $t(s_i^4)$) 
is chosen, then this choice forces to connect the horizontal segments' 
end-points vertically aligned with it.
Thus, we get a conflict-free assignment of the variables. 
\end{proof}

\begin{lemma}
 The \textsc{Max 2-Sat} instance $\psi$ satisfies at least $k$ clauses simultaneously if and only if 
 the weight of ${\cal T}$ is at most  $W({\cal T}_{\cal S}) + (m-k)\varepsilon$, where 
 $m$ is the number of clauses in $\psi$, $\varepsilon>0$ is a 
 very small real number, and $1 \leq k < m$.
\end{lemma}

\begin{proof}
{(\bf Necessity)} Let $\psi$ satisfies at least $k$ clauses simultaneously.
For a literal $x_{i,j,k}$ (or $\neg x_{i,j,k}$), if variable $x_i$ is true (or false),
then $T$ (or $F$) is chosen accordingly; this implies, one of the 
end-points of the segment $s^{\lceil\frac{j}{2}\rceil}$ 
can be chosen to connect with ${\cal T}_{\cal S}$ in MSTS ${\cal T}$ without 
adding any extra weight (see Figure \ref{fig:2sat_msts}, where the blue edge 
implies $\neg x_1$ is {\it true}, i.e., $x_1$ is {\it false}).
Therefore, if at least one literal is true in one 
clause, no extra weight will be added to 
${\cal T}_{\cal S}$ for that clause to get $W({\cal T})$. If both the 
literals are false, then the extra weight of $\varepsilon$
will be added to $W({\cal T}_{\cal S})$. Therefore, if 
$\psi$ satisfies at least $k$ clauses simultaneously, then an extra weight of at most
$(m-k)\times \varepsilon$ will be added to $W({\cal T}_{\cal S})$ 
(see the red edge in Figure \ref{fig:2sat_msts}).
\begin{figure}[!h]
 \centering
 \includegraphics[scale=0.85]{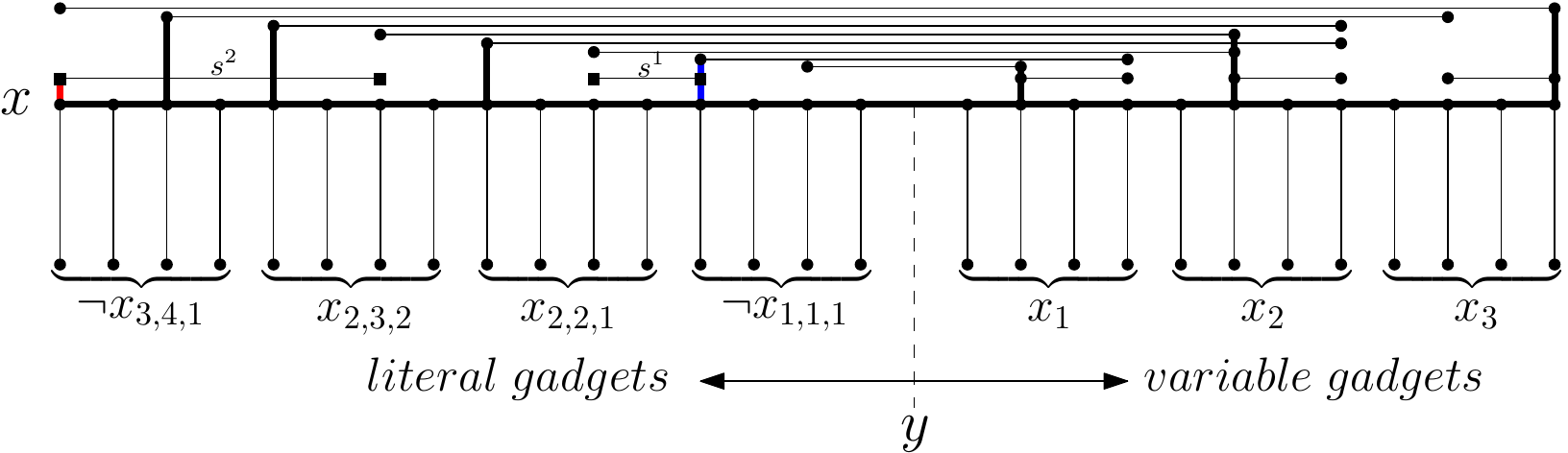}
 \caption{The MSTS (shown in thick) obtained for the assignment 
 $x_1={\it false},x_2={\it false}, \text{ and } x_3={\it true}$}\label{fig:2sat_msts}
\end{figure}

\noindent {(\bf Sufficiency)} Let ${\cal T}$ be an MSTS over 
${\cal S}'$ such that 
$W({\cal T}) \leq W({\cal T}_{\cal S}) + (m-k)\varepsilon$. 
We show that there is a truth assignment of variables such that at least $k$ clauses are satisfied. In any 
gadget, all the top end-points of its vertical segments are part of ${\cal T}$. 
Let in ${\cal T}$, the horizontal segments of literal and variable gadgets are connected to 
vertical segments in one of the ways as shown in 
Figure \ref{fig:gad_mst} (if not, we can alter the tree as discussed 
in the proof of Lemma \ref{lem:msts_assign}). 

By Lemma \ref{lem:msts_assign}, it is guaranteed that we can obtain a 
conflict-free assignment to the variables. Now consider the second term (i.e., $(m - k)\varepsilon$) 
of $W({\cal T})$. Each factor $\varepsilon$ is due to the non-existence of an
edge between some segment $s^{\lceil\frac{j}{2}\rceil}$ end-points and any one of the horizontal 
segments' end-points of the $j$-th and $(j+1)$-st literal 
gadgets. We assign {\it false} to both the literals in the clause associated with 
the segment $s^{\lceil\frac{j}{2}\rceil}$, implying the clause 
corresponding to $s^{\lceil\frac{j}{2}\rceil}$ is not satisfied. 
Since at most $m-k$ clauses are not satisfied, the result follows.
\end{proof}

Given a designated set of end-points (one marked end-point of each segment) of $\cal S$
and a parameter $\mu$, in polynomial-time it can be decided whether the sum of lengths of the 
edges of the MST of those points is less than or equal to $\mu$ by simply computing an MST. 
Thus, we have the following result:

\begin{theorem}\label{thm:msts_hard}
The decision version of the MSTS problem is NP-complete.
\end{theorem}

\subsection{Implication of the reduction}\label{sub_sec:implication}
\begin{enumerate}
 \item The MSTS problem remains NP-hard even for the case where all the segments are horizontal. Here, the same reduction works 
by replacing each vertical segment by a segment of length zero.

\item By a simple modification to the horizontal segments in the reduction in Section \ref{sec:msts_nphard} can be used to show the unrestricted version of the MSTS problem, we call it MIN-MSTS problem, is NP-hard. The MIN-MSTS problem is formally defined below. 
\begin{description}
 \item[The MIN-MSTS problem]
 \item[Instance] A set $\cal S$ of $n$ disjoint segments.
 \item[Objective] Choose one point form each segment such that the cost of the MST of the chosen points is as small as possible.
\end{description}

We achieve the NP-hardness of the MIN-MSTS by extending all the horizontal segments by a small unit, say $\frac{1}{2}$, from both the end-points, see Figure \ref{fig:2sat_msts_mod}. In the new arrangement of gadgets, the MST chooses a point from each horizontal segment, which was an end-point in the previous reduction, and the cost of the MST is same as the MST obtained for the MSTS problem. 
\begin{figure}[!h]
 \centering
 \includegraphics[scale=0.85]{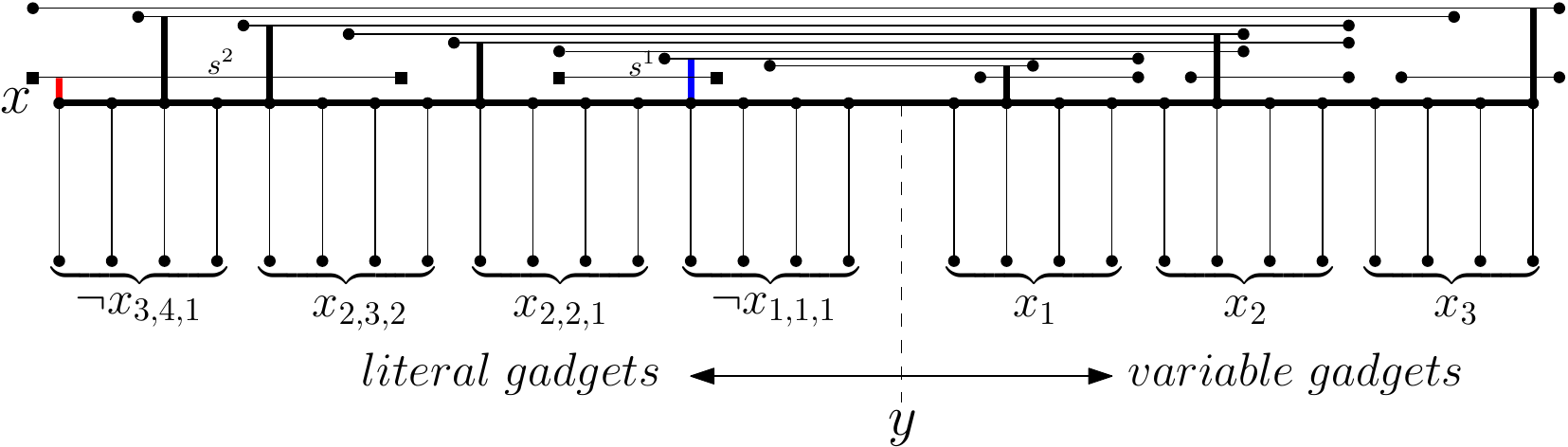}
 \caption{Modified reduction}\label{fig:2sat_msts_mod}
\end{figure}

\end{enumerate}
\begin{theorem}
 The MIN-MSTS problem is NP-hard.
\end{theorem}

\section{Parameterized algorithm for the MSTS problem} \label{sec:para_algo}
In this section, we propose a parameterized algorithm, where the parameter $k$ is the \emph{separability} 
of the given instance. In \cite{dorrigiv2015minimum}, the concept of ``separability'' was introduced to propose a $(1 + \frac{2}{k})$-factor 
algorithm for the MST problem defined on disjoint disks. We observe that using the similar concept it is possible to obtain a $(1+\frac{1}{k})$-factor 
approximation for the MST problem on segments. We now define separability for segments.

Let $\ell_{max}$ be the maximum length of any segment in the given instance. We say that the 
given instance is \emph{$k$-separable} if the minimum distance between any two segments is 
at least $k\cdot \ell_{max}$. Note that there might be many values of $k$ 
satisfying the above distance requirement. The maximum $k$ for which the given instance is 
$k$-separable is called as the \emph{separability} of the given instance.

Suppose we are given a set ${\cal S} = \{s_1,s_2,\ldots,s_n\}$ of $n$ disjoint line segments that satisfies $k$-separability.
In our algorithm, for each segment $s_i$, we arbitrarily pick one of it's end-points, and compute an MST on the chosen set of 
end-points. Let $T$ and $T_{opt}$ be the obtained MST on the set of chosen points and an optimal MSTS\footnote{which also chooses an appropriate end-point from each 
segment such that it's cost is as minimum as possible} on $\cal S$, respectively. Needless to say that $cost(T_{opt}) \leq cost (T)$. 
Also, let $T'$ be a spanning tree which has the same topology\footnote{Its nodes are points used for $T$, but the connections are as follows: each edge $(s_i,s_j)$ 
(an end-point of $s_i$ is connected with an end-point of $s_j$) in $T_{opt}$ is present among the chosen end-points of $s_i$ and $s_j$ in $T'$.} as $T_{opt}$ but on the points used for $T$. Observe that 
$cost(T)\leq cost(T')$ as $T$ is an MST and $T'$ is a spanning tree on the same set of points. Consider an arbitrary edge 
$e'$ in $T'$. Let $s_i$ and $s_j$ be the segments of lengths $\ell_i$ and $\ell_j$, respectively, such that $e'=(p_i,p_j)$ is 
defined on the end-point of $p_i$ of $s_i$ and the end-point $p_j$ of $s_j$.
 If $d$ is the (smallest) 
distance between $s_i$ and $s_j$, then $cost(e')$ and the cost of the edge $e$ in $T_{opt}$ connecting $s_i$ and $s_j$ is greater than or equal to $d$. 

If $p_i'$ is the other end-point of $s_i$, then by triangle inequality we have $cost(e') \leq length(s_i) + length(p_j,p_i')$. Therefore,
\begin{equation*}
 \frac{cost(e)}{cost(e')} \geq \frac{d}{length(s_i) + length(p_j,p_i')} \geq \frac{k\cdot \ell_{max}}{\ell_{max} + k\cdot \ell_{max}}\geq \frac{k}{1+k},
\end{equation*}
thus $cost(e') \leq (1 + \frac{1}{k})\cdot cost(e)$. As this inequality is true for every edge of $T'$, we have the following result due to the 
fact that $cost(T) \leq cost(T')$.

\begin{theorem}\label{thm:para}
For the MSTS problem on disjoint segments with separability parameter $k > 0$, 
the algorithm that builds an MST by choosing an arbitrary end-point from each segment achieves an approximation factor of $(1+ \frac{1}{k})$.
\end{theorem}

Theorem \ref{thm:para} suggests that if the value of $k$ is very large, then the solution produced by the algorithm is very close to the optimal. This happens only when the segments are far apart from each other.
\section{Approximation Algorithm}\label{sec:approx} 
In this section, we propose a constant factor approximation algorithm for the MSTS problem for 
a given set of $n$ non-crossing segments ${\cal S} = \{s_1,s_2,\ldots,s_n\}$. Let  
the end-points of $s_i$ be $p_i$ and $q_i$, i.e., $s_i=\overline{p_iq_i}$. Also, let $o_i$ be a hypothetical\footnote{The point $o_i$ is basically not present; it is used for the 
graph-theoretic formulation of the problem.} point associated with each 
segment $s_i \in S$, see Figure \ref{fig:apx}(a). We construct a graph $G=(V,E)$
with $3n$ nodes such that $V=\{p_i,q_i,o_i\mid i=1,2,\ldots,n\}$. The node $o_i$ is connected with the 
end-points,  $p_i$ and $q_i$, of 
its corresponding segment $s_i$. The edges in $E$ are as follows: first draw a 
complete graph with the nodes $\{p_i,q_i \mid i=1,2,\ldots,n\}$, and then remove the edges 
$\{(p_i,q_i)\mid i=1,2,\ldots,n\}$. For each $i\neq j$, the weight of the edges $(p_i,p_j)$, 
$(q_i,q_j)$, $(p_i,q_j)$, and $(p_j,q_i)$ are equal to their corresponding Euclidean distances. For 
each $i$, the weight of the edges $(p_i,o_i)$ and $(q_i,o_i)$ is 0 (zero)\footnote{Note that, $o_i$ can be reached only from $p_i$ and $q_i$, but with cost 0}. The nodes $\{o_i\mid i=1,\ldots, n\}$ are considered to be the \emph{terminal} nodes, and $\{p_i,q_i \mid
i=1,2,\ldots, n\}$ are \emph{non-terminal} nodes. 

\begin{obs}\label{obs}
If ${\cal T}_{st}$ is a minimum cost Steiner tree for spanning the terminal nodes in $G$, 
then the cost an MSTS for the set of segments in $\cal S$ is at least
$cost({\cal T}_{st})$. 
\end{obs}
\begin{proof}
If ${\cal T}_{opt}$ is an optimum MSTS, then we can get a feasible solution to the Steiner tree problem by adding a terminal node $o_i$ to an end-point of it's associate segment $s_i=\overline{p_iq_i}$ that participates in ${\cal T}_{opt}$. Thus, $cost({\cal T}_{st}) \leq cost({\cal T}_{opt})$.
\end{proof}

Now, we will propose a simple two-pass algorithm to compute a spanning tree for the segments in $S$. In the first pass, we execute the algorithm to compute 
a Steiner tree $\cal T$ for the graph $G$ with cost at most $\alpha\times {\cal T}_{st}$, where $\alpha$ is the
best-known approximation factor available in the literature for the Steiner tree problem.   
Note that, $\cal T$ may include both the end-points of some segment(s). If no such event is observed, report it as the result by removing the nodes $\{o_i\mid i=1, \ldots, n\}$ and the edges incident to those nodes. Otherwise, we execute the second pass 
to modify $\cal T$ to obtain another Steiner tree such that it contains exactly one end-point from each segment in $S$. 

\begin{enumerate}
\item Start from a non-terminal node $r$, called the {\it root}, and traverse the tree in pre-order (root and then all its children counterclockwise order\footnote{This ordering can easily be computed by considering the slopes of the half-lines originating at the root and coinciding with the edges incident at the root}). While traversing the tree, the segment whose both endpoints are present in 
$\cal T$ are pushed in a stack while it is recognized. These segments will be referred to as {\it bad} segments, and the end-points of a bad segment are referred to as \emph{bad end-points} (respectively {\it bad vertex} in $G$).

\item Consider the bad segments one by one from the stack, i.e., each time one bad segment is popped out from the stack for processing. It's bad end-point, say $a$, that is not connected with its corresponding terminal vertex, is deleted. Let $\chi = \{c_1, c_2, \ldots, c_k\}$ be the neighbors of $a$ in $\cal T$ in counterclockwise order. Now, connect $c_i,c_{i+1}$ for $i = 1,\ldots,k-1$, see Figure \ref{fig:apx}(d). %

\item The process in Step 2 terminates when the stack is empty.
\end{enumerate}
Let ${\cal T}'$ be the resulting tree after execution of Step 3 (i.e., after processing the last element in the stack, for example, see Figure \ref{fig:apx}(d) and Figure \ref{fig:apx_ana}(d)). The resulting tree ${\cal T}'$ is (i) a Steiner tree as it still connects every terminal node, and (ii) it contains exactly one end-point (i.e., non-terminal node) from each segment in $\cal S$. 

\begin{obs}
 A feasible solution to the MSTS problem can be obtained by removing all the terminal nodes $o_i$  from ${\cal T}'$, and it's cost is $cost({\cal T}')$.
\end{obs}

The following lemma describes the relation between the costs of the trees $\cal T$ and ${\cal T}'$. 
\begin{lemma} \label{lem:2_approx}
$cost({\cal T}') \leq 2\times cost({\cal T})$
\end{lemma}
\begin{proof}
Let $a$ be a bad vertex of $G$ which is an end-point of a bad segment $s_i$, where $s_i$ is the top element of the stack at the beginning of the second pass. If $\{c_1, c_2, \ldots, c_k\}$ are 
the nodes adjacent with $a$ in $\cal T$, then in the revised tree, the newly introduced edges satisfy $cost(c_i,c_{i+1}) \leq cost(c_i,a) + cost(c_{i+1},a)$, $i=1,2,\ldots, k-1$ (by triangle inequality). Thus, the cost of every edge $(c_i,a)$ is counted twice and each edge $(c_i,c_{i+1})$ is counted once for $i=1,2,\ldots, k-1$. Hence the sum of costs of the newly added edges is at most twice the cost of the removed edges from $T$. 
Let $s_j$ be the next element in the stack (if any), and let $b$ be the bad end-point considered in Step 2. Only the following situations may take place.

\begin{figure}[!h]
    \centering
    \subfigure[]{\includegraphics[width=0.35\textwidth]{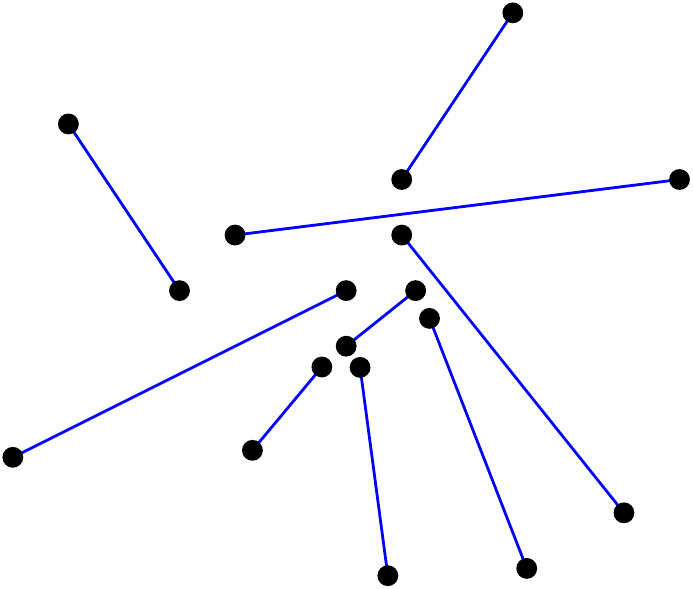}} \hspace*{0.5in}
    \subfigure[]{\includegraphics[width=0.35\textwidth]{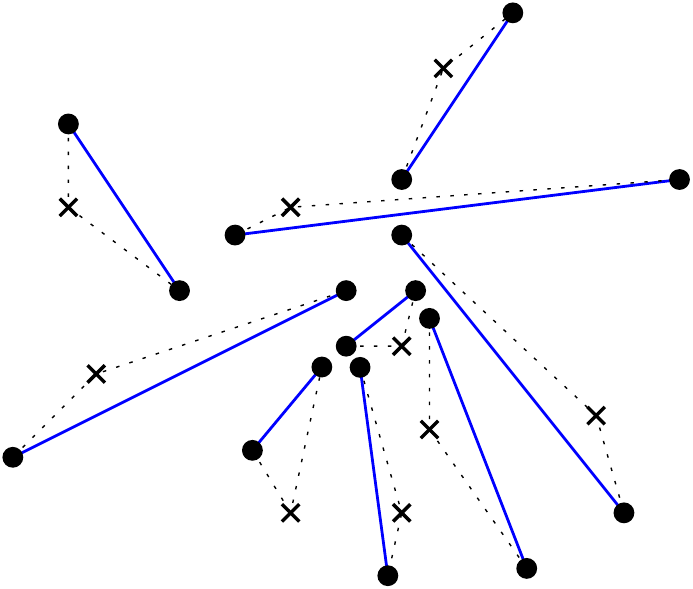}} \\
    \subfigure[]{\includegraphics[width=0.35\textwidth]{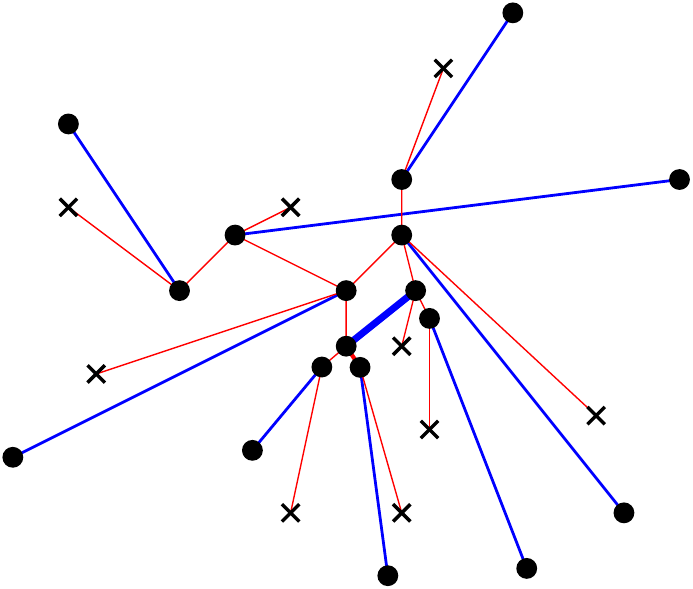}}\hspace*{0.5in}
    \subfigure[]{\includegraphics[width=0.35\textwidth]{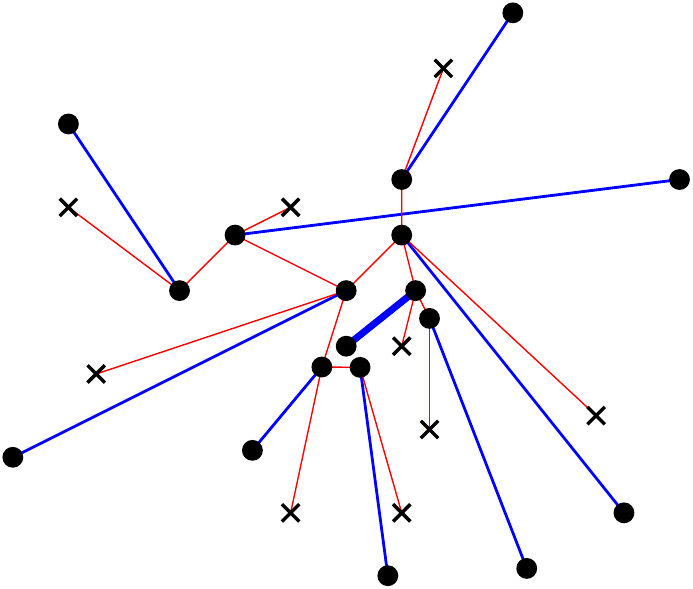}}
    \caption{(a) A set of line segments in $I\!\!R^2$, (b) Terminal nodes, shown using cross marks, associated with each segment, (c) A Steiner tree which has both the end-points of a {\it bad} segment (showed in thick), and (d) A modified Steiner tree: we can obtain a spanning tree for the segments by deleting the terminal nodes}
    \label{fig:apx}
\end{figure}

\begin{description}
 \item[Case 1:] The bad end-point $b$ of $s_j$ is adjacent only to non-bad
 vertices in $\cal T$, see Figure \ref{fig:apx}(c) and Figure \ref{fig:apx}(d). This situation is handled as earlier.
  \item[Case 2:] The bad end-point $b$ of $s_j$ is adjacent with an already 
  processed bad vertex (say $a$) of $\cal T$, see \ref{fig:apx_ana}(c). The other neighbors of $b$ are $\{d_1, d_2, \ldots\}$. Note that, none of $d_i$s are same as any of $c_i$'s; otherwise it contradicts the fact that $\cal T$ is a tree. As earlier, we delete $b$ from the resulting tree obtained in the previous iteration and connect the neighbors of $b$ in counterclockwise order to obtain a tree for the next iteration. Observe that one of $c_i$s will be adjacent to one of $d_j$s.
  \begin{figure}
   \subfigure[]{\includegraphics[scale=0.75]{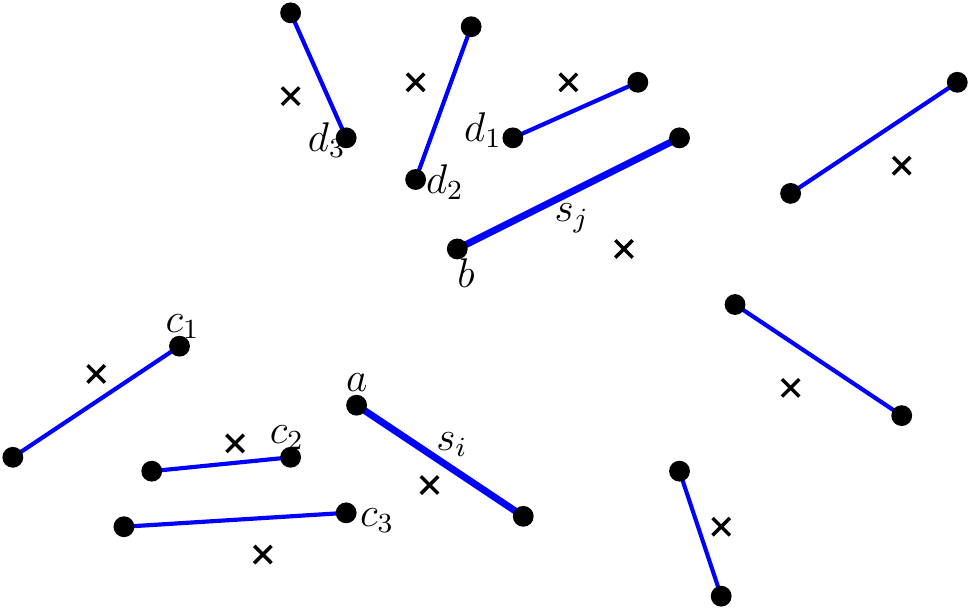}}\hspace*{0.5in}
   \subfigure[]{\includegraphics[scale=0.75]{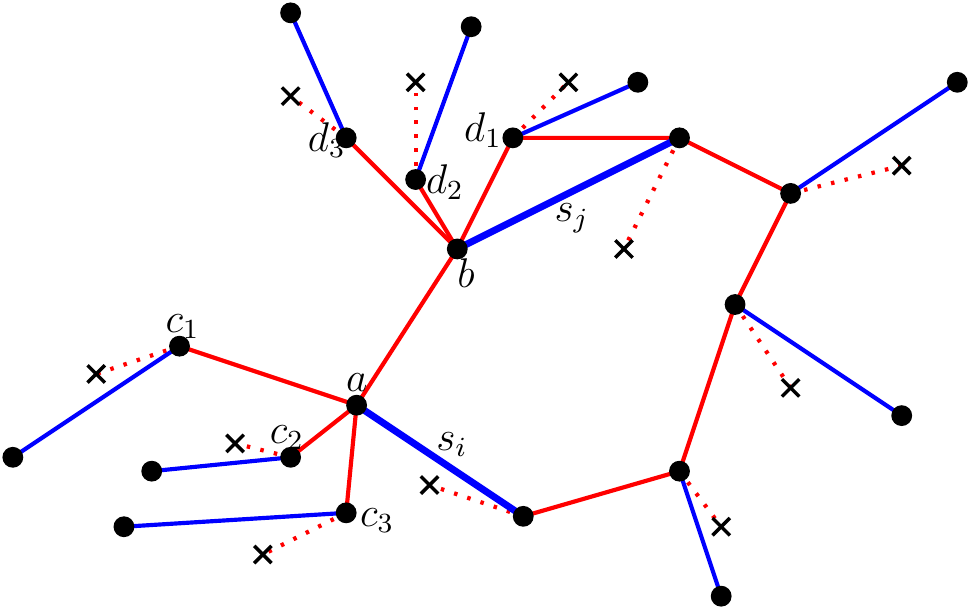}}\\
   \subfigure[]{\includegraphics[scale=0.75]{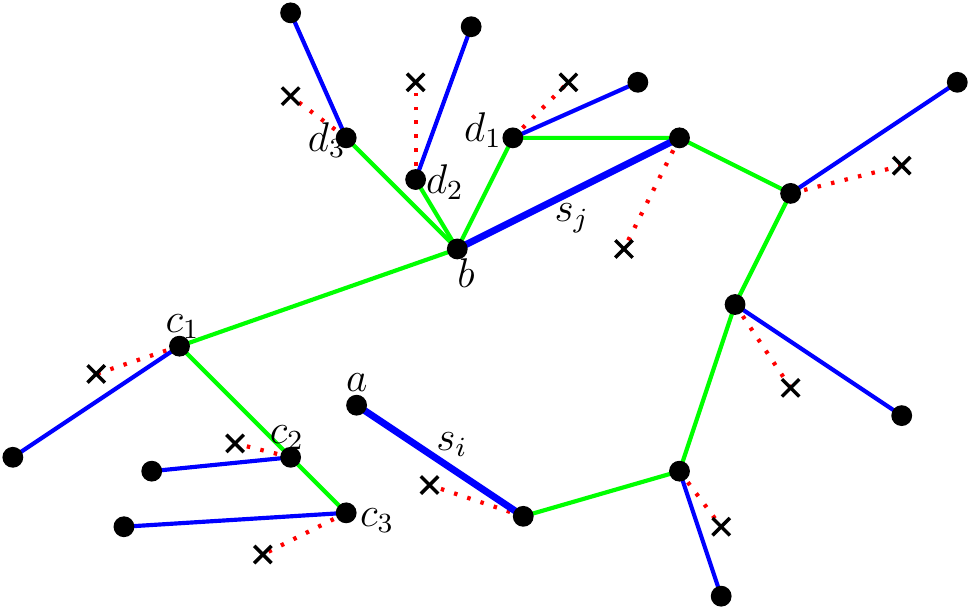}}\hspace*{0.5in}
   \subfigure[]{\includegraphics[scale=0.75]{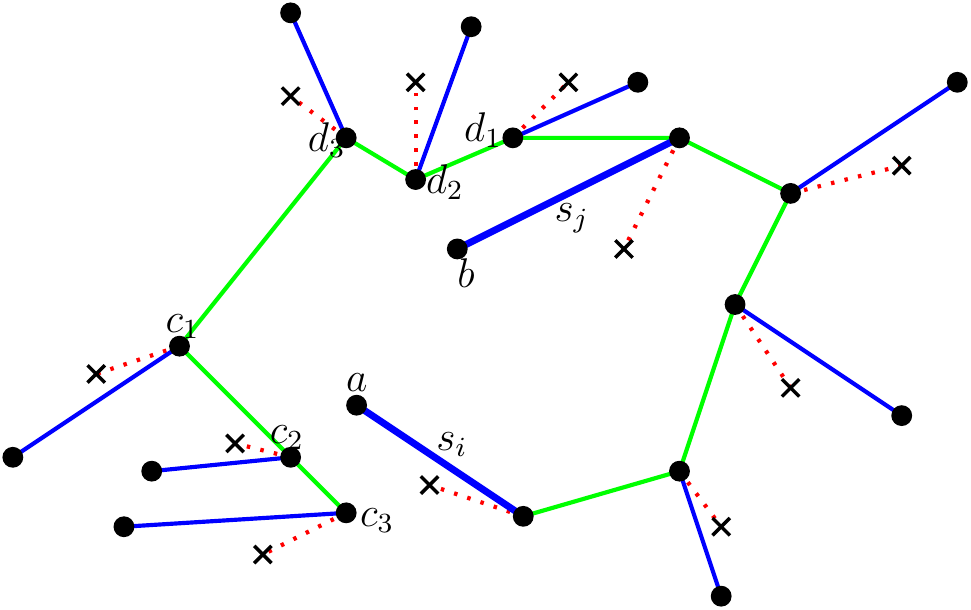}}
   \caption{(a) An instance to the STP with terminal nodes (showed in cross), (b) A Steiner tree spanning all the terminals, (c) The resulting tree after deletion of the bad vertex $a$, and (d) The resulting tree after deletion of the bad vertex $b$ which was adjacent to $a$ in ${\cal T}$. Initially, the stack contains $s_j$ and $s_i$ with $s_i$ as it's topmost element. We first process the segment $s_i$ and then $s_j$. The resulting trees after processing the segments $s_i$ and $s_j$ are shown in (c) and (d), respectively.}\label{fig:apx_ana}
  \end{figure}
Hence, $cost(c_i,d_j) \leq cost(c_i,b) + cost(b,d_j) \leq cost(c_i,a) +cost(a,b) + cost(b,d_j)$. Thus, in ${\cal T} \setminus \{(c_i,a)\}$ (i.e., after considering the deletion of $a$), the cost of the edge $(a,b)$ is accounted for (deleted) only once, the cost of the edge $(c_i,a)$ is accounted for (deleted) only twice (for deletion of $b$ in this step), and no other edges in $\{(c_i,a)\mid i=1, 2, \ldots \}$ is considered in this step. Note that, $(a,b)$ was  deleted once from $\cal T$ while deleting the vertex $a$, and is considered while deleting the vertex $b$ (in the above inequality)   as $a$ is a neighbor of $b$ in ${\cal T}$, and it will never be considered for deletion in this phase 
further. Same argument follows for processing other elements of the stack and the above invariant is ensured all the time.
\end{description}
Finally, adding the inequalities, the resulting sum of lengths of the edges of ${\cal T}'$ is at most twice the sum of lengths of the edges of $\cal T$. 
\end{proof}

We construct a tree, say ${\cal T}''$, from ${\cal T}'$ by deleting all the terminal nodes from ${\cal T}''$, and we have the following result.

\begin{theorem}
Let ${\cal S} = \{s_1,s_2,\ldots,s_n\}$ be a set of $n$ non-crossing line segments in the plane. If $T_{opt}$ be the optimal
MSTS, then $cost({\cal T}'') \leq 2\alpha \times cost(T_{opt})$, where $\alpha$ is the best-known approximation factor for 
the minimum Steiner tree problem for an undirected graph.  
\end{theorem}

\begin{proof}
Observe that (i) $cost({\cal T}'') = cost({\cal T}')$, and (ii) ${\cal T}''$ is a feasible solution for the MSTS problem as it contains exactly one end-point form each segment $s_i \in {\cal S}$. By Lemma \ref{lem:2_approx}, we get $cost({\cal T}') \leq 2 \times cost({\cal T}) \leq 2\alpha \times cost({\cal T}_{st})$. Again, 
$cost({\cal T}_{st}) \leq cost(T_{opt})$ due to Observation \ref{obs}. Thus, we have $cost({\cal T}'') \leq 2\alpha \times cost({\cal T}_{opt})$.
\end{proof}
To the best of our knowledge, the value of $\alpha$ is $\ln 4 + \varepsilon < 1.39$, where $\varepsilon$ is a positive constant \cite{byrka2010improved}. Thus, we claim that 
the MSTS problem admits a 2.78-factor approximation algorithm. However, the proposed algorithm in \cite{byrka2010improved} is an LP-based iterative randomized rounding technique (later they derandomize it) in which an LP is solved at every iteration after contracting a component. Recently, Chen and Hsieh \cite{chen2020efficient} claim to propose an efficient two-phase heuristic in greedy strategy that achieves an approximation ratio of 1.4295. By using the algorithm in \cite{chen2020efficient} in the first phase of our algorithm, we can speedup our algorithm with scarifying approximation factor a little bit.

\section{Conclusion}
We considered the Euclidean minimum spanning problem for the non-crossing line segments neighborhoods, and we showed that the problem is NP-hard in general. For the restricted MSTS problem, we proposed (i) a parametrized approximation algorithm based on the separability parameter defined for segments, and (ii) a $2\alpha$-factor approximation algorithm which uses the best-known $\alpha$-factor approximation algorithm as a subroutine to compute a minimum cost Steiner tree in an undirected edge weighted graph. Another point worth concluding here is that our algorithm can also be used to solve a variant of the GMST problem considered in \cite{pop}, and our algorithm produces a solution whose cost is at most 2.78 times to the optimum. This is a notable improvement over the 4-factor 
approximation proposed in \cite{pop} and it is the best-known approximation available in the literature. Improving the approximation factor for the MSTS problem may be a problem of further interest.  

\bibliographystyle{plain}
\bibliography{msts}

\end{document}